\documentclass[a4paper,12pt]{article}
\usepackage{amsmath}
\usepackage{amsthm}
\usepackage{amssymb}
\usepackage{mathrsfs}
\usepackage{graphics}
\usepackage{cite}
\usepackage{indentfirst}
\usepackage{braket}
\usepackage{hyperref}
\hypersetup{hidelinks,colorlinks=false}

\theoremstyle{definition}

\newtheorem{remark}{Remark}
\theoremstyle{plain}
\newtheorem{theorem}{Theorem}

\DeclareMathOperator{\Tr}{Tr} 

\renewcommand*{\geq}{\geqslant}

\numberwithin{equation}{section}

\begin{document}

\title{\Large\textbf{Derivation of the Redfield quantum master equation and corrections to it by the Bogoliubov method}}
\author{Anton Trushechkin\footnote{e-mail: trushechkin@mi-ras.ru}}
\date{\textit{\small Steklov Mathematical Institute of Russian Academy of Sciences, Gubkina St. 8, 119991 Moscow, Russia}}

\maketitle

\begin{abstract}
Following the ideas N. N. Bogoliubov used to derive the classical and quantum nonlinear kinetic equations, we give an alternative derivation of the Redfield quantum linear master equation, which is widely used in the theory of open quantum systems, as well as higher-order corrections to it. This derivation naturally considers initially correlated system-reservoir states arising from the previous system-reservoir dynamics. It turns out that the Redfield equation does not require any modifications in this case. The expressions of higher-order corrections are simpler than those obtained by other methods.
\end{abstract}

\section{Introduction}

The theory of open quantum systems deals with quantum systems interacting with the environment (reservoir). The state of an isolated finite-dimensional quantum system is defined by the density operator, i.e., a positive Hermitian operator $\rho$ with unit trace \cite{Holevo, OhyaVol}. The dynamics of the density operator of an isolated quantum system is described by the von Neumann equation (quantum Liouville equation). If the system interacts with the environment, it is generally impossible to write a closed differential equation for the density operator of the system. However, in some limiting cases, such as the limit of weak coupling between the system and reservoir, the limit of low density of particles in the reservoir, and some others, one can derive closed linear differential equations for the density operator of the system, i.e., quantum master equations \cite{Redfield,BP}; moreover, this can be done in a mathematically rigorous way \cite{Davies,AccLuVol,Huelga,MerkliRev}. There are several systematic methods for deriving quantum master equations: the Zwanzig projection operator method \cite{Davies,Zwanzig,Huelga}, the quantum stochastic limit method \cite{AccLuVol}, and the resonance theory \cite{MerkliRev}.

In this study, we try to apply the method that N. N. Bogoliubov used when deriving the classical and quantum nonlinear kinetic equations for a gas (the Boltzmann, Landau, and Vlasov equations and their quantum analogs)  \cite{Bogol46,BogolGurov} (see also a recent review \cite{Geras}) to the derivation of quantum linear master equations. Studying the weak-coupling regime between the system and reservoir, we will provide an alternative derivation of the well-known Redfield equation \cite{Redfield,BP} and higher-order corrections to it with respect to the coupling constant between the system and reservoir. 

The underlying perturbation theory is developed from the von Neumann equation directly, rather then from the formally exact Nakajima--Zwanzig master equation, and adopts the concept of a kinetic state and the kinetic stage of evolution: During the initial fast relaxation process, the joint system-reservoir state approaches a kinetic (adjusted) state, i.e., a state form a certain class where the joint system-reservoir state depends on time only through the reduces state of the system and, thus, is completely determined by it. In the zeroth order, a kinetic state corresponds to a product (uncorrelated) system-reservoir state, but higher orders give correlations: Obviously, the states developing from the interacting system-reservoir dynamics are correlated. The kinetic stage of evolution corresponds to evolution of the kinetic states.

If the initial state is not a product state (assumed in the common derivations of quantum master equations), but a correlated kinetic state originated from the previous dynamics, then there is no initial relaxation process and we have the kinetic evolution from the very beginning. Thus, the correlated initial states arising from the previous dynamics are naturally considered in the presented approach and, moreover, simplify the analysis in comparison with the product states. The assumption of a product initial state is often criticized to be unrealistic in many cases \cite{MerkliCorr,VacchiniCorr,Pechukas,Tasaki,GVF}  (see also a discussion in Remark~\ref{Rem} below) precisely because of the previous system-reservoir dynamics. 

Recently (after a journal publication of this paper), new mathematical tools to take arbitrary initial correlations into account in the formalism of master equations were proposed \cite{MerkliCorr,VacchiniCorr}. Namely, in Ref.~\cite{MerkliCorr}, correction terms caused by initial correlations are rigorously estimated, while the method of Ref.~\cite{VacchiniCorr} uses a decomposition of an arbitrary system-reservoir state into a sum of products of system operators and reservoir states (proposed in Ref.~\cite{Hall}). 

However, correlations caused by the previous system-reservoir dynamics are not arbitrary, but very special. So, if the initial correlations are of this particular form, then the presented formalism is simpler since such initial correlations are even more natural here than the uncorrelated states and require no corrections. 

In particular, it turns out that the Redfield equation does not require any modifications in the case of a kinetic initial state. Moreover, it is better suited for precisely such kind of initial conditions: Well-known problem of possible violation of positivity on the initial short times is caused by a fast non-Markovian relaxation process of the initial product state to a correlated kinetic (adjusted) system-reservoir state \cite{Slip}. So, we have no violation of positivity if we start with a kinetic initial state.

Corrections to the Redfield equation have been studied in a number of papers \cite{4thorder,TrushHigh,BP,RichterMuka}. However, in these papers, corrections to the Redfield master equation explicitly depend on time; i.e., the corresponding differential equations are nonautonomous. The time dependence of the corrections reflects the aforementioned fast relaxation processes of the initial product state to a kinetic system-bath state in a short initial time interval. If we start with a kinetic system-reservoir state, then there is no initial relaxation process and, as a consequence, the quantum master equation turns out to be an autonomous differential equation in all orders and has a simpler form. 

If, nevertheless, the initial state is exactly a product state, then the presented formalism also can be used in combination with the usual time-dependent perturbation theory, which allows one to calculate the dynamics on initial short times. Thus, the general philosophy in this case is to use the time-dependent perturbation theory (or another method) for an initial short time period and then to use a simpler master equation for the kinetic stage of the dynamics, rather then to use a more complicated master equation describing both the initial fast relaxation to a kinetic state and the kinetic stage of the dynamics.

Also note that the perturbation theory based on the averaging method \cite{GVF,BMP,FLP} also separates a Markovian master equation for the kinetic stage and non-Markovian short-term corrections due to the initial fast relaxation. It also starts directly with the von Neumann equation bypassing the Nakajima--Zwanzig master equation. Also, in Ref.~\cite{GVF}, another special type of initial correlations is considered: namely, equilibrium system-reservoir state with external field switching-off at the zero time instant. But the concept of a kinetic state was not formulated within this method.

Since a kinetic system-reservoir state is completely determined by the reduced state of the system, the presented method allows us to calculate not only the density operator of the system at every instant of time but also the joint state of the system and reservoir. Note that the state of the reservoir and the joint state of the system and reservoir can be studied also using the Zwanzig projection operator method \cite{TrushZwan, TrushJCP}, the quantum stochastic limit method, and the resonance theory. Corrections to the quantum stochastic limit were studied in Ref.~\cite{PechenVol}.

When deriving quantum master equations, we will follow the method described in Chapter 9 of Bogoliubov's monograph \cite{Bogol46}. Let us explain the difference in the statement of the problem. Describing the dynamics of a gas of a large number of identical particles in terms of a single-particle state leads to nonlinear master equations: nonlinear terms are responsible for the interaction of particles (each of which is described by a single-particle state at a current instant of time) with each other. In the theory of open quantum systems, we deal with a single quantum system interacting with a large reservoir. This statement of the problem leads to linear master equations for the state of the system (see Refs.~\cite{Spohn,SpohnRev} for more details on this difference). Just as Bogoliubov's derivation of nonlinear master equations, our derivation of linear quantum master equations for an open quantum system is performed at the ``physical'' level of rigor.

\section{Problem statement}

Consider a quantum system interacting with a reservoir, and let $\mathcal H_S$ be a finite-dimensional Hilbert space corresponding to the system and $\mathcal H_R$ be an infinite-dimensional Hilbert space corresponding to the reservoir. Denote by $I_R$ the identity operator in the space $\mathcal H_R$. Then an arbitrary operator $A$ in the space of the system can be naturally embedded in the space $\mathcal H_S\otimes\mathcal H_R$ as $A\otimes I_R$; the operators in the space of the reservoir can also be embedded in $\mathcal H_S\otimes\mathcal H_R$ in a similar way. In the space $\mathcal H_S\otimes\mathcal H_R$, consider a Hamiltonian (self-adjoint operator) of the form

\begin{equation}
H=H_S+H_R+\lambda H_I=H_0+\lambda H_I,
\end{equation}
where $H_S$ is the Hamiltonian of the isolated system, $H_R$ is the Hamiltonian of the isolated reservoir, $H_0=H_S+H_R$ is the Hamiltonian of free evolution, $H_I$ is the Hamiltonian of interaction between the system and reservoir, and $\lambda$ is a small dimensionless parameter. Suppose that $H_I$ is a finite sum of the form

\begin{equation}
H_I=\sum_\alpha T_\alpha\otimes B_\alpha,
\end{equation}
where $T_\alpha$ are operators in the space of the system and $B_\alpha$ are operators in the space of the reservoir.

Define the corresponding Liouville operators acting on spaces of operators (i.e., superoperators):  $\mathcal L_0=[H_0,\,\cdot\,]$, $\mathcal L_I=[H_I,\,\cdot\,]$, and $\mathcal L_S=[H_S,\,\cdot\,]$, where the square brackets stand for the commutator of two operators. Introduce the unitary evolution superoperators $\mathcal U_0(t)O=e^{-iH_0t}Oe^{iH_0t}\equiv O(-t)$ and
$\mathcal U_S(t)O=e^{-iH_St}Oe^{iH_St}$, where $O$ is an arbitrary linear operator in $\mathcal H_S\otimes\mathcal H_R$. In addition, define
\begin{equation}
\mathcal L_I(t)=[H_I(t),\,\cdot\,]=\mathcal U_0(-t)\mathcal L\,\mathcal U_0(t).
\end{equation}

Denote by $\rho=\rho(t)$ the joint state of the system and reservoir at time $t$. Since the space of the reservoir is infinite-dimensional, not every state can be represented as a density operator: in the general case, a state is an identity-preserving linear functional on an algebra of observables \cite{OhyaVol}. As the algebra of observables, we take the algebra $\mathfrak A$ generated by operators of the form $T\otimes I_R$ and $T\otimes B_\alpha(t)$ for all possible operators $T$ in the space of the system, indices $\alpha$, and time instants $t\in\mathbb R$. Then, $\rho$ is a nonnegative linear functional on this algebra, its value on the identity of the algebra (on the identity operator) being equal to one. Nevertheless, in the notation, we treat $\rho$ as a density operator (in a ``generalized'' sense). For example, by an expression $\Tr\rho A$ we mean the value of the functional $\rho$ on an element $A$ of the algebra $\mathfrak A$.

Denote by $\rho_S(t)=\Tr_R\rho(t)$ the reduced density operator of the system, where $\Tr_R$ is the partial trace over the space $\mathcal H_R$; i.e., by definition, $\rho_S(t)$ is a density operator such that $\Tr\rho_S(t)T=\Tr[\rho(t)(T\otimes I_R)]$ for all operators $T$ in the space $\mathcal H_S$.

The state $\rho(t)$ satisfies the von Neumann equation

\begin{equation}\label{EqLiuv}
\dot\rho=-i\mathcal L_0\rho-i\lambda\mathcal L_I\rho;
\end{equation}
hence,
\begin{equation}\label{EqLiuvS}
\dot\rho_S=-i\mathcal L_S\rho_S-i\lambda\Tr_R(\mathcal L_I\rho).
\end{equation}

Let us state the problem of deriving a quantum master equation of the form

\begin{equation}\label{EqMaster}
\dot\rho_S=\mathcal G\rho_S,
\end{equation}
where $\mathcal G$ (from ``generator'') is a linear superoperator in the space of the system, i.e., an operator acting in the space of operators in $\mathcal H_S$.

Equation (\ref{EqLiuvS}) is not a closed equation for $\rho_S$. Hence, to obtain an equation of the form  (\ref{EqMaster}), we should consider a particular solution of equation (\ref{EqLiuv}) whose total state  $\rho$ at every time is completely determined by the reduced density operator of the system $\rho_S$:

\begin{equation}\label{EqR}
\rho(t)=\mathcal R\rho_S(t),
\end{equation}
where $\mathcal R$ (from ``recovery'') is a linear operator acting from the space of operators in the space of the system to the space of joint states of the system and reservoir. We can call it a recovery operator: the state of the system and reservoir is recovered from the reduced density operator of the system. It is also called the assignment map \cite{Pechukas,AlickiPechukas,Sudar,Guzik}.

Thus, $\rho(t)$ depends on time only through $\rho_S(t)$. Of course, this cannot be a general solution, because then the initial condition $\rho(0)$ would not be arbitrary but rather would be determined by the form of  $\rho_S(0)$. Nevertheless, one can expect from physical considerations (fast relaxation processes in the reservoir) that the state $\rho(t)$ rapidly approaches a state of the form (\ref{EqR}) under general physically admissible initial conditions. Assumption (\ref{EqR}) is an analog of the assumption in Bogoliubov's method that the multiparticle probability density depends on time only through the single-particle probability density. System-reservoir state of the class (\ref{EqR}) are called the kinetic states since only for such states a kinetic (master) equation can be formulated. The stage of evolution after relaxation of an initial state to a kinetic state, i.e., evolution of the states of form (\ref{EqR}), is called the kinetic stage of evolution. We stress that condition (\ref{EqR}) is required to hold only on the algebra $\mathfrak A$, i.e., on a fairly restricted set of observables.

The partial trace over the space of the reservoir is the operation inverse to $\mathcal R$; therefore, we require that the following equality should hold (consistency property):
\begin{equation}\label{EqTrRR}
\Tr_R\mathcal R\rho_S=\rho_S.
\end{equation}
Since the map $\mathcal R$ is linear, different from the product map (see Eq.~(\ref{EqR0}) below), and satisfies the consistency property (\ref{EqTrRR}), it cannot be positive unless we restrict its domain \cite{Pechukas,AlickiPechukas,Sudar,Guzik}. This could be expected: Not all reduced states of the system $\rho_S$ can be achieved after the aforementioned fast relaxation. We will consider $\mathcal R$ acting on all operators $\rho_S$, but we will refer to the subset of $\rho_S$ mapped to positive density operators as the positivity domain of $\mathcal R$ \cite{Sudar}.

Denote by $\Phi(t)$ the semigroup of operators acting in the space of linear operators on $\mathcal H_S$ that is generated by $\mathcal G$. If $\rho(0)=\mathcal R\rho_S(0)$, then 
\begin{equation}\label{EqSemiGr}
\rho_S(t)=\Phi(t)\rho_S(0).
\end{equation} 

If we substitute a formal solution of the von Neumann equation  (\ref{EqLiuv}),
$$\rho(t)=\mathcal U(t)\rho(0)=\mathcal U(t)\mathcal R\rho_S(0),$$
into the left-hand side of Eq.~(\ref{EqR})  and substitute a formal solution of the master equation (\ref{EqSemiGr}) into
the right-hand side, then, since $\rho_S(0)$ is arbitrary, we obtain
\begin{equation}\label{EqCommut}
\mathcal U(t)\mathcal R=\mathcal R\Phi(t).
\end{equation}

\section{Expansion in powers of the small parameters}

Thus, we address the problem of determining the expressions on the right-hand sides of Eqs.~(\ref{EqMaster}) and (\ref{EqR}) for which $\rho(t)$ satisfies Eq.~(\ref{EqLiuv}). To solve the problem, we employ expansions in the small parameter $\lambda$; namely, we will choose the expansion coefficients so that the expression

\begin{equation}\label{EqRSer}
\mathcal R\rho_S=\mathcal R_0\rho_S+\lambda\mathcal R_1\rho_S+\lambda^2\mathcal R_2\rho_S+\ldots
\end{equation}
with $\rho_S$ found from the equation
\begin{equation}\label{EqGser}
\dot\rho_S=\mathcal G_0\rho_S+\lambda \mathcal G_1\rho_S+\lambda^2 \mathcal G_2\rho_S+\ldots
\end{equation}
formally satisfies the von Neumann equation (\ref{EqLiuv}). Note that all $\mathcal R_r$ and $\mathcal G_r$ should be linear operators.

We expect that the truncated generator $\mathcal G_0+\ldots+\lambda^n\mathcal G_n$ is positive whenever $\lambda$ is small enough and $\rho_S$ belongs to the positivity domain of $\mathcal R$. The Redfield master equation is known to violate positivity on initial short times for some initial $\rho_S$. From the point of view of the presented formalism, positivity is violated on $\rho_S$ outside the positivity domain of $\mathcal R$. This agrees with the common view on possible initial violation of positivity by the Redfield equation \cite{Slip}: The system and reservoir adjust to each other's state on the initial fast time scale, on which the dynamics is highly non-Markovian and not described by the Redfield (or another Markovian) master equation. Generally, it is well-known that complete positivity can be violated in the case of initial system-reservoir correlations  \cite{Pechukas,Sudar,Lidar}.

It follows from Eq.~(\ref{EqTrRR}) that
\begin{equation}\label{EqTRRr}
\Tr_R\mathcal R_0\rho_S=\rho_S\qquad\text{and}\qquad
\Tr_R\mathcal R_r\rho_S=0,\quad r\geq1.
\end{equation}
Taking the partial trace over the reservoir in Eq.~(\ref{EqLiuv}), we immediately see that
\begin{equation}\label{EqG}
\mathcal G_0=-i\mathcal L_S\qquad\text{and}\qquad \mathcal G_r=-i\Tr_R[\mathcal L_I\mathcal R_{r-1}(\cdot)],\quad r\geq1.
\end{equation}
Next, differentiating equation (\ref{EqRSer}) and substituting expression (\ref{EqGser}) for $\dot\rho_S$, we obtain
\begin{equation}\label{EqDotRhoSSer}
\begin{split}
\mathcal R\dot\rho_S&=\mathcal R_0\mathcal G_0\rho_S+\lambda \mathcal R_0\mathcal G_1\rho_S+\lambda^2\mathcal R_0\mathcal G_2\rho_S+\ldots\\
&+\lambda \mathcal R_1\mathcal G_0\rho_S+
\lambda^2\mathcal R_1\mathcal G_1\rho_S+
\lambda^3\mathcal R_1\mathcal G_2\rho_S+\ldots\\
&+\lambda^2 \mathcal R_2\mathcal G_0\rho_S+
\lambda^3\mathcal R_2\mathcal G_1\rho_S+
\lambda^4\mathcal R_2\mathcal G_2\rho_S+\ldots\\&=
\sum_{r=0}^\infty\lambda^r\sum_{n=0}^r\mathcal R_{r-n}\mathcal G_n\rho_S.
\end{split}
\end{equation}
Substituting $\mathcal R\rho_S(t)$ for $\rho(t)$ in Eq.~(\ref{EqLiuv}), then replacing $\mathcal R\rho_S(t)$ and $\mathcal R\dot\rho_S(t)$ by their expansions (\ref{EqRSer})  and (\ref{EqDotRhoSSer}), respectively, and equating the coefficients of the same powers of $\lambda$, we obtain
\begin{equation}\label{EqGeneric}
\sum_{n=0}^r\mathcal R_{r-n}\mathcal G_n\rho_S=-i\mathcal L_0\mathcal R_r\rho_S-i\mathcal L_I\mathcal R_{r-1}\rho_S,\quad r\geq 0,
\end{equation}
where we set $\mathcal R_{-1}\equiv0$ by definition.

\section{Boundary conditions and recurrence formula}

To solve equations (\ref{EqGeneric}) for $\mathcal R_r$, one should specify the corresponding boundary conditions. To
this end, we proceed from the assumption that under free dynamics a joint state of the form (\ref{EqR}) relaxes to a state of the form $\mathcal U_S(t)\rho_S\otimes\rho_R^{\rm ref}$, where $\rho_R^{\rm ref}$ is a certain ``reference'' state of the reservoir RR
that is stationary with respect to its free dynamics: $e^{-iH_Rt}\rho_R^{\rm ref}e^{iH_Rt}=\rho_R^{\rm ref}$. For example, this may be a Gibbs state at some temperature. That is,
\begin{equation}
\mathcal U_0(\tau)\mathcal R\rho_S-
[\mathcal U_S(\tau)\rho_S]\otimes\rho_R^{\rm ref}\to0,\quad \tau\to+\infty,
\end{equation}
or
\begin{equation}\label{EqBoundCond}
\mathcal U_0(\tau)\mathcal R\mathcal U_S(-\tau)\rho_S\to
\rho_S\otimes\rho_R^{\rm ref},\quad \tau\to+\infty.
\end{equation}
These conditions should be satisfied on the algebra $\mathfrak A$. They are similar to the conditions of weakening of correlations between distant particles in the derivation of the Boltzmann equation by the Bogoliubov method.

For the terms of the series expansion, the boundary conditions have the form
\begin{equation}\label{EqBound}
\begin{split}
\mathcal U_0(\tau)\mathcal R_0\,
\mathcal U_S(-\tau)\rho_S&\to\rho_S\otimes\rho_R^{\rm ref},\\
\mathcal U_0(\tau)\mathcal R_r\,\mathcal U_S(-\tau)\rho_S&\to0,\quad r\geq 1,
\end{split}
\end{equation}
as $\tau\to+\infty$.

The following formula holds:
\begin{equation}\label{EqRrG0}
\mathcal R_r\mathcal G_0\mathcal U_S(\tau)\rho_S=
-i\mathcal R_r\mathcal L_S\mathcal U_S(\tau)\rho_S=
\frac{d}{d\tau}\mathcal R_r\mathcal U_S(\tau)\rho_S,
\end{equation}
where $\tau$ is an arbitrary real number. Here, the first equality holds in view of Eq.~(\ref{EqG}), and the second, in view of the fact that $-i\mathcal L_S$ is the generator of the unitary group $\mathcal U_S$.

Replacing $\rho_S$ in Eq.~(\ref{EqGeneric})  by $\mathcal U_S(\tau)\rho_S$ and applying formula (\ref{EqRrG0}), we obtain
\begin{equation}\label{EqGeneric2}
\frac{d}{d\tau}[\mathcal R_r\mathcal U_S(\tau)\rho_S]=
-i\mathcal L_0[\mathcal R_r\mathcal U_S(\tau)\rho_S]
-i\mathcal L_I\mathcal R_{r-1}\mathcal U_S(\tau)\rho_S
-\sum_{n=1}^r\mathcal R_{r-n}\mathcal G_n\mathcal U_S(\tau)\rho_S,\quad r\geq 0.
\end{equation}
We will regard this equality as an inhomogeneous equation for  $\mathcal R_r\mathcal U_S(\tau)\rho_S$. To stress this fact, we marked this operator in the equation with square brackets.

When $r=0$, the inhomogeneous part is absent and the equation takes the form

\begin{equation}
\frac{d}{d\tau}[\mathcal R_0\mathcal U_S(\tau)\rho_S]=
-i\mathcal L_0[\mathcal R_0\mathcal U_S(\tau)\rho_S];
\end{equation}
hence, 
\begin{equation}
\mathcal R_0\mathcal U_S(\tau)\rho_S=\mathcal U_0(\tau)\mathcal R_0\rho_S.
\end{equation}
Replacing the arbitrary argument $\rho_S$ by $\mathcal U_S(-\tau)\rho_S$, we obtain
\begin{equation}
\mathcal R_0\rho_S=\mathcal U_0(\tau)\mathcal R_0\mathcal U_S(-\tau)\rho_S.
\end{equation}
Since this equality holds for an arbitrary $\tau$ and its left-hand side is independent of $\tau$, we can pass
to the limit as $\tau\to+\infty$ on the right-hand side and, using Eq.~(\ref{EqBound}),  write

\begin{equation}\label{EqR0}
\mathcal R_0\rho_S=\rho_S\otimes\rho_R^{\rm ref}.
\end{equation}
So, in the limit $\lambda\to0$, the recovery operator maps an arbitrary $\rho_S$ into a product state with the reference state of the reservoir. This agrees with rigorous proofs that the product state is justified in the weak-coupling limit \cite{Tasaki,Yuasa,MerkliCorr}.

\begin{remark}\label{Rem}
We can see that  
$
\mathcal P_0\rho=\mathcal R_0(\Tr_R\rho)
$
is the standard projection operator that is conventionally used to derive a quantum master equation in the weak-coupling regime. Therefore, the expansion (\ref{EqRSer}) can be viewed as a generalization of the projection operator method. Namely, define the operator
\begin{equation}\label{EqP}
\mathcal P\rho=\mathcal R(\Tr_R\rho)=
\mathcal P_0\rho+\lambda \mathcal P_1\rho+\lambda^2\mathcal P_2\rho+\ldots,
\end{equation}
where $\mathcal P_r=\mathcal R_r(\Tr_R\rho)$. Conditions (\ref{EqTrRR}) and (\ref{EqTRRr}) guarantee that $\mathcal P$ is a projector, i.e., $\mathcal P^2=\mathcal P$. A state of the form (\ref{EqRSer}) onto which $\mathcal P$ projects provides a more accurate description of the state of the system and reservoir arising after fast relaxation than a state of the form (\ref{EqR0}). Let us explain this.

Under free dynamics, the state of the system and reservoir relaxes to Eq.~(\ref{EqR0}). However, the interaction constantly ``pushes out'' the state from the form  (\ref{EqR0}). Therefore, when the interaction is turned on, the state of the system and reservoir relaxes not to a state of the form (\ref{EqR0}) but rather to a state of the form (\ref{EqRSer}), which contains corrections to  Eq.~(\ref{EqR0}) with respect to the small parameter. This state corresponds to the ``equilibrium'' between the free relaxation dynamics in the reservoir, which brings the state to the form (\ref{EqR0}),  and the interaction that pushes it out from this form. 

If the initial state is not of the form  (\ref{EqRSer})  but of the form (\ref{EqR0}), then an initial non-Markovian period of the dynamics arises, which is described by nonperturbative terms and is noted and analyzed in the literature (see, for example, Refs.~\cite{Slip,Taepra}). It is this period that corresponds to the relaxation of a state to the form (\ref{EqRSer}). If the initial state has the form (\ref{EqRSer}), then no such terms arise and the dynamics is Markovian from the very beginning.

Which initial state is ``more correct,''  (\ref{EqR0}) or (\ref{EqRSer})? State  (\ref{EqR0}) corresponds to the situation when there has been no interaction (therefore, the system and reservoir have been in the state of absence of correlations), which is abruptly turned on at zero time, while  state (\ref{EqRSer}) corresponds to the situation when interaction has existed for a long time and we start to observe the system at some time instant.

The results that we are going to obtain can also be established by the projection operator method if we adopt projection operator (\ref{EqP}) with initially unknown components $\mathcal P_1,\mathcal P_2,\ldots$ and find these components by a recurrent procedure. It is interesting to compare our approach with R.~Balescu's approach to the derivation of nonlinear master equations of a system of interacting particles in classical statistical mechanics \cite{Balescu}. In this approach, it is required that the projection operator $\mathcal P$ should exactly satisfy the condition
\begin{equation}\label{EqPU}
\mathcal U(t)\mathcal P=\mathcal P\mathcal U(t).
\end{equation}
Then one constructs the projection operator $\mathcal P$ as a perturbation series satisfying this condition.
Let us show that our operator (\ref{EqP}) satisfies the weaker condition
\begin{equation}\label{EqPUP}
\mathcal U(t)\mathcal P=\mathcal P\mathcal U(t)\mathcal P;
\end{equation}
i.e., the subspace onto which $\mathcal P$ projects (following Balescu's book, we can call it a kinetic subspace) is invariant under evolution. The fulfillment of condition (\ref{EqPUP}) follows from the fact that a state of the form  (\ref{EqR}) is a solution of the Liouville equation; i.e., if the initial state has the form $\mathcal R\rho_S(0)$ (belongs to the kinetic subspace), then the state at time $t$, i.e., $\mathcal R\rho_S(t)$, also belongs to the kinetic subspace. We can also present a formal chain of equalities for an arbitrary state $\rho$: 
\begin{equation*}
\begin{split}
\mathcal P\mathcal U(t)\mathcal P\rho&=
\mathcal R\Tr_R\mathcal U(t)\mathcal R\Tr_R\rho\\&=
\mathcal R\Tr_R\mathcal R\Phi(t)\Tr_R\rho\\&=
\mathcal R\Phi(t)\Tr_R\rho\\&=
\mathcal U(t)\mathcal R\Tr_R\rho\\&=
\mathcal U(t)\mathcal P\rho,
\end{split}
\end{equation*}
where we have used equalities (\ref{EqCommut}) and (\ref{EqTrRR}). To achieve the stronger equality  (\ref{EqPU}), it is necessary and sufficient that we have the additional equality $\mathcal P\mathcal U(t)=\mathcal P\mathcal U(t)\mathcal P$, or, equivalently, $\mathcal P\mathcal U(t)\mathcal Q=0$, where $\mathcal Q$ is the projector onto the orthogonal complement of the kinetic subspace (nonkinetic subspace). This equality means that even if the initial state contains a nonkinetic part, it does not affect the dynamics of the kinetic part. In view of the above-mentioned results of Ref.~\cite{Taepra} on the emergence of nonperturbative terms, this assertion fails to hold: the nonkinetic part (for example, if the initial state has the form (\ref{EqR0}), which, as said above, contains a nonkinetic part when terms of order higher than zero are taken into account) affects the dynamics of the kinetic part.
\end{remark}

Let us return to the solution of Eq.~(\ref{EqGeneric2}) for $r\geq1$. If expressions for $\mathcal R_0,\mathcal R_1,\ldots,\mathcal R_{r-1}$ are already known, then we also know the inhomogeneous part. Therefore, the solution can be written as
\begin{equation}
\mathcal R_r\mathcal U_S(\tau)\rho_S=\mathcal U_0(\tau)\mathcal R_r\rho_S
-\int_0^\tau ds\,\mathcal U_0(s)\Big\lbrace
i\mathcal L_I\mathcal R_{r-1}+
\sum_{n=1}^r\mathcal R_{r-n}\mathcal G_n\Big\rbrace
\mathcal U_S(\tau-s)\rho_S.
\end{equation}
Replacing the arbitrary argument $\rho_S$ by $\mathcal U_S(-\tau)\rho_S$, we obtain
\begin{equation}
\mathcal R_r\rho_S=\mathcal U_0(\tau)\mathcal R_r\mathcal U_S(-\tau)\rho_S
-\int_0^\tau ds\,\mathcal U_0(s)\Big\lbrace
i\mathcal L_I\mathcal R_{r-1}+
\sum_{n=1}^r\mathcal R_{r-n}\mathcal G_n\Big\rbrace
\mathcal U_S(-s)\rho_S.
\end{equation}
Again passing to the limit as $\tau\to+\infty$ and applying the boundary conditions (\ref{EqBound}), we find
\begin{equation}\label{EqRecur}
\begin{split}
\mathcal R_r\rho_S&=
-\int_0^\infty ds\,\mathcal U_0(s)\Big\lbrace
i\mathcal L_I\mathcal R_{r-1}+
\sum_{n=1}^r\mathcal R_{r-n}\mathcal G_n\Big\rbrace
\mathcal U_S(-s)\rho_S\\
&=-\int_0^\infty ds\,\Big\lbrace
i\mathcal L_I(-s)\mathcal R_{r-1}(-s)+
\sum_{n=1}^r\mathcal R_{r-n}(-s)\mathcal G_n(-s)\Big\rbrace
\rho_S,
\end{split}
\end{equation}
where $\mathcal R_r(t)=\mathcal U_0(-t)\mathcal R_r\,\mathcal U_S(t)$ and $\mathcal G(t)=\mathcal U_S(-t)\mathcal G\,\mathcal U_S(t)$.

Let us check that equalities  (\ref{EqTRRr}) hold for $r\geq1$. For $r=1$, the equality is straightforward in
view of formula (\ref{EqG}) for $\mathcal G_r$. For $r>1$, one establishes Eq.~(\ref{EqTRRr}) by induction, also using Eq.~(\ref{EqG}).

Formula (\ref{EqRecur}) is a recurrence formula that allows us to successively determine all $\mathcal R_r$ and,
hence, all $\mathcal G_r$:
\begin{equation}\label{EqGr}
\mathcal G_r\rho_S=-\int_0^\infty ds\,
\Tr_R\Big\lbrace
\mathcal L_I\mathcal L_I(-s)\mathcal R_{r-2}(-s)\rho_S
-i
\sum_{n=1}^{r-1}\mathcal L_I\mathcal R_{r-n-1}(-s)\mathcal G_n(-s)\rho_S\Big\rbrace.
\end{equation}
In the interaction representation, the generator takes the form
\begin{equation}\label{EqGrInt}
\mathcal G_r\rho^I_S=-\int_0^\infty ds\,
\Tr_R\Big\lbrace
\mathcal L_I(t)\mathcal L_I(t-s)\mathcal R_{r-2}(t-s)\rho^I_S
-i
\sum_{n=1}^{r-1}\mathcal L_I(t) \mathcal R_{r-n-1}(t-s)
\mathcal G_n(t-s)\rho^I_S\Big\rbrace,
\end{equation}
where $\rho_S^I(t)=\mathcal U_S(-t)\rho_S(t)$ is the reduced density operator of the system in the interaction representation.

Note that in the interaction representation the differential equation for the reduced density operator of the system, $\rho_S^I(t)$, is no longer autonomous. However, we can recover the autonomous differential equation by applying the secular approximation \cite{BP,Huelga}, which consists in eliminating the rapidly oscillating terms from Eq.~(\ref{EqGrInt}), or a more general approximation \cite{TrushUni}.

As mentioned above, if $\rho(0)\neq\mathcal R\rho_S(0)$, then the generator $\mathcal G$ allows us to calculate the dynamics on times longer than the short initial time of relaxation of the joint state of the system and reservoir to the form (\ref{EqRSer}). The question arises as to how to calculate the dynamics in this short initial stage. Even if we are not interested in the details of the transient dynamics itself, we need to know the state of the system to which this dynamics leads, because it is this state that should serve as the initial condition for the master equation (\ref{EqMaster}). The problem of correct choice of the initial states for quantum master equations was studied in Ref.~\cite{Taepra}. One of the possible variants is to take the first several terms of the chronological exponential
\begin{equation}\label{EqTexp}
\begin{split}
\rho(t_0)&=T_+\exp\left\{-i\int_0^{t_0}ds\,\mathcal L_I(s)\right\}\rho(0)
\\
&\equiv
\rho(0)-i\int_0^{t_0}ds\,\mathcal L_I(s)\rho(0)+
(-i)^2\int_0^{t_0}ds_1\int_0^{s_1}ds_2\,\mathcal L_I(s_1)\mathcal L(s_2)\rho(0)
+\ldots
\end{split}
\end{equation}
and a time  $t_0$ such that the approximation in the form of several terms of the series (\ref{EqTexp}) is yet adequate, while the relaxation to a state of the form (\ref{EqRSer}) has already occurred. Then the calculated value of $\rho(t_0)$ can be taken as the initial value for the master equation (\ref{EqMaster}).

\section{Well-definiteness of the expressions for $\mathcal R_r$ and $\mathcal G_r$}

Expressions (\ref{EqRecur}) and (\ref{EqGr}) have been derived at the ``physical'' level of rigor; however, we will
show that under certain conditions these expressions are well defined. Throughout the rest of the
paper, we assume that the reference state of the reservoir $\rho_R^{\rm ref}$ and operators $B_\alpha$ appearing in $H_I$ have the following properties:
\begin{align}
&\langle
B_{\alpha_1}(t_1)\cdots
B_{\alpha_{2n+1}}(t_{2n+1})
\rangle_R=0,\label{EqOddZero}\\
&\langle
B_{\alpha_1}(t_{1})\cdots B_{\alpha_{2n}}(t_{2n})
\rangle_R=
\sum 
\prod_{j=1}^n
\langle
B(t_{j_1})B(t_{j_2})
\rangle_R\label{EqEven}
\end{align}
for any indices $\alpha_j$ and time instants $t_j$. Here $\langle A\rangle_R=\Tr_R A\rho_R$ for any $A\in\mathfrak A$. The sum in Eq.~(\ref{EqEven}) is taken over all partitions of the set $\{1,\ldots,2n\}$ into $n$ pairs  $\{(j_1,j_2)\}$ such that $j_1<j_2$. Then, in view of the equality
$$\langle B_\alpha(t)B_\beta(s)\rangle_R=
\langle B_\alpha(t-s)B_\beta\rangle_R,$$
the generators of all orders can be expressed in terms of pair correlation functions
\begin{equation}\label{EqC}
C_{\alpha\beta}(t)=\langle B_\alpha(t)B_\beta\rangle_R=
\chi_{\alpha\beta}(t)-i\varphi_{\alpha\beta}(t),
\end{equation}
where $\chi_{\alpha\beta}(t)$ and $\varphi_{\alpha\beta}(t)$ are real functions. Note that we have the equality $C_{\alpha\beta}(-t)=C_{\beta\alpha}^*(t)$.

\begin{theorem}
Suppose that conditions (\ref{EqOddZero}) and (\ref{EqEven}) are satisfied and the correlation functions $C_{\alpha\beta}(t)$ decay rapidly, i.e., $t^nC_{\alpha\beta}(t)\to0$ as $t\to\infty$ for all $n$. Then expressions (\ref{EqRecur}) and (\ref{EqGr}) are well defined for all $r$.
\end{theorem}
\begin{proof}
It suffices to prove that $\mathcal R_r\rho_S$ are well-defined linear functionals on the algebra $\mathfrak A$, since the fact that $\mathcal G_r\rho_S$ are well defined follows from here in view of Eq.~(\ref{EqG}).

We can easily prove by induction that $\mathcal R_r\rho_S$ is a sum of terms of the form
\begin{equation}
\int_0^\infty ds_1
\ldots
\int_0^\infty ds_k
\,
\mathcal L(-s_1)
\mathcal L(-s_1-s_2)
\cdots
\mathcal L(-s_1-\ldots-s_k)
\mathcal T(s_1,\ldots,s_k)\rho_S\otimes\rho_R^{\rm ref},
\end{equation}
where $\mathcal T(s_1,\ldots,s_k)$  is a superoperator in the space of the system that depends on $s_1,\ldots,s_k$ (product of different $\mathcal G_r$). Hence, using Eqs.~(\ref{EqOddZero}) and (\ref{EqEven}), we can easily see that the value of the functional $\mathcal R_r\rho_S$ on an arbitrary element of the algebra $\mathfrak A$ is the integral of the product of rapidly decreasing functions of the variables $s_1,\ldots,s_k$, which is convergent.
\end{proof}

\section{Redfield equation and the first correction to it}

Let us obtain the first two nonzero terms of the expansion of the generator (\ref{EqGser}) in an explicit form. By  Eqs.~(\ref{EqG})   and (\ref{EqOddZero}), we conclude that $\mathcal G_1\equiv 0$. Let us derive an expression for $\mathcal G_2$. Applying equation (\ref{EqRecur}) for $r=1$, we get
\begin{equation}\label{EqR1}
\mathcal R_1\rho_S=
-i\int_0^\infty ds\,\mathcal L_I(-s)(\rho_S\otimes\rho_R^{\rm ref}).
\end{equation}

Hence, 
\begin{equation}
\mathcal G_2\rho_S=-\int_0^\infty ds\,\Tr_R\left[\mathcal L_I\mathcal L_I(-s)(\rho_S\otimes\rho_R^{\rm ref})\right].
\end{equation}
This is the standard Redfield generator. In the interaction representation, the generator takes the form
\begin{equation}\label{EqG2int}
\mathcal G_2(\rho^I_S)=-\int_0^\infty ds\,\Tr_R\left[\mathcal L_I(t)\mathcal L_I(t-s)(\rho_S^I\otimes\rho_R^{\rm ref})\right].
\end{equation}
It is well known that this generator turns into a generator of the Gorini--Kossakowski--Lindblad--Sudarshan (GKLS)\footnote{We follow the order of the names proposed in Ref.~\cite{BriefGKLS}.} form if one applies the secular approximation  or a more general approximation.

Applying equation (\ref{EqRecur}) for $r=2$, we obtain
\begin{multline}\label{EqR2}
\mathcal R_2\rho_S=-\int_0^\infty ds_2\int_0^\infty ds_1
\Big\lbrace
\mathcal L_I(-s_2)\mathcal L_I(-s_1-s_2)(\rho_S\otimes\rho_R^{\rm ref})\\-
\Tr_R
\big[
\mathcal L_I(-s_2)\mathcal L_I(-s_1-s_2)(\rho_S\otimes\rho_R^{\rm ref})
\big]\otimes\rho_R^{\rm ref}
\Big\rbrace.
\end{multline}

In view of  Eq.~(\ref{EqOddZero}), we have $\mathcal G_3\equiv0$. To find $\mathcal G_4$, we need to get an expression for $\mathcal R_3\rho_S$. Applying
equation (\ref{EqRecur}) for $r=3$, we have
\begin{multline}\label{EqR3}
\mathcal R_3\rho_S=
-i\int_0^\infty ds_3\int_0^\infty ds_2\int_0^\infty ds_1\\\Big\{
\mathcal L_I(-s_1-s_3)
\Tr_R\big[\mathcal L_I(-s_3)\mathcal L_I(-s_2-s_3)
(\rho_S\otimes\rho_R^{\rm ref})
\big]\otimes\rho_R^{\rm ref}\\-
\mathcal L_I(-s_3)\mathcal L_I(-s_2-s_3)\mathcal L_I(-s_1-s_2-s_3)
(\rho_S\otimes\rho_R^{\rm ref})\\+
\mathcal L_I(-s_3)
\Tr_R\big[\mathcal L_I(-s_2-s_3)\mathcal L_I(-s_1-s_2-s_3)
(\rho_S\otimes\rho_R^{\rm ref})\big]\otimes\rho_R^{\rm ref}
\Big\}
\end{multline}
and
\begin{equation}\label{EqG4}
\begin{split}
\mathcal G_4\rho_S=
\int_0^\infty ds_3\int_0^\infty ds_2\int_0^\infty ds_1
\Big\{
&\langle\mathcal L_I\mathcal L_I(-s_3)\mathcal L_I(-s_2-s_3)\mathcal L_I(-s_1-s_2-s_3)
\rangle_R
\\-
&\langle\mathcal L_I\mathcal L_I(-s_3)\rangle_R
\langle
\mathcal L_I(-s_2-s_3)\mathcal L_I(-s_1-s_2-s_3)\rangle_R
\\-
&\langle\mathcal L_I\mathcal L_I(-s_1-s_3)\rangle_R
\langle\mathcal L_I(-s_3)\mathcal L_I(-s_2-s_3)\rangle_R
\Big\}.
\end{split}
\end{equation}
In the interaction representation, the generator takes the form
\begin{equation}\label{EqG4int}
\begin{split}
\mathcal G_4(\rho^I_S)=
\int_0^\infty ds_3\int_0^\infty ds_2\int_0^\infty ds_1
\Big\{
&\langle\mathcal L_I(t)\mathcal L_I(t-s_3)\mathcal L_I(t-s_2-s_3)
\mathcal L_I(t-s_1-s_2-s_3)
\rangle_R
\\-
&\langle\mathcal L_I(t)\mathcal L_I(t-s_3)\rangle_R
\langle
\mathcal L_I(t-s_2-s_3)\mathcal L_I(t-s_1-s_2-s_3)\rangle_R
\\-
&\langle\mathcal L_I(t)\mathcal L_I(t-s_1-s_3)\rangle_R
\langle\mathcal L_I(t-s_3)\mathcal L_I(t-s_2-s_3)\rangle_R
\Big\}.
\end{split}
\end{equation}

Let us calculate the averages over the reservoir in the expressions for $\mathcal G_2$ and $\mathcal G_4$ in the interaction representation in the case when $H_I=T\otimes B$ (i.e., when $H_I$ contains a single term). In this case, we can omit the subscripts of the correlation function $C(t)$ and its components $\chi(t)$ and $\varphi(t)$ (see Eq.~(\ref{EqC})).
Introduce the superoperators
\begin{equation*}
\mathcal C(t)=[T(t),\,\cdot\,],\quad
\mathcal D(s,t)=\chi(s)[T(t-s),\,\cdot\,]-i\varphi(s)\{T(t-s),\,\cdot\,\},
\end{equation*}
where $\{\cdot,\cdot\}$ is the anticommutator. Then a direct calculation shows that
\begin{equation}
\mathcal G_2\rho_S=
-\int_0^\infty ds\,\mathcal C(s)\mathcal D(s,t)\rho_S
\end{equation}
and
\begin{multline}\label{EqG4intCD}
\mathcal G_4\rho_S=
\int_0^\infty ds_1
\int_0^\infty ds_2
\int_0^\infty ds_3\\
\big\lbrace
\mathcal C(s)\mathcal C(t-s_3)
\big[
\mathcal D(s_2+s_3,t)
\mathcal D(s_1+s_2,t-s_3)
+
\mathcal D(s_2,t-s_3)
\mathcal D(s_1+s_2+s_3,t)
\big]
\\
-
\mathcal C(t)\mathcal D(s_1+s_3,t)
\mathcal C(t-s_3)\mathcal D(s_2,t-s_3)
\big\rbrace\rho_S.
\end{multline}

Expressions (\ref{EqG4int}) and (\ref{EqG4intCD}) for the correction $\mathcal G_4$ to the Redfield generator have a simpler form than the corresponding expressions in Ref.~\cite{4thorder} (see also Ref.~\cite{TrushHigh}), in which the initial transient dynamics from a product system-reservoir state (\ref{EqR0}) to a kinetic state is taken into account. As we said above, in our approach one can consider this dynamics separately, after which one can apply a simpler master equation. If the system and reservoir interacted before and arrive at the initial time instant in a kinetic state $\mathcal R\rho_S(0)$, then there is no initial transient dynamics and the obtained expressions (\ref{EqG4int}) and (\ref{EqG4intCD}) can be used from the very beginning.

Note that the generator of the $2n$th order provides the accuracy of a steady-state solution $\rho_S^*$ up to the $(2n-2)$th order for the part of $\rho_S^*$ commuting with $H_S$ (i.e., the diagonal part) and up to the $2n$th order for the part of $\rho_S^*$ non-commuting with $H_S$ \cite{FlemingCummings,Purkayastha}.

\section{Is the quantum dynamics Markovian?}

The fact that the reduced density operator of the system satisfies the closed ordinary differential equation (\ref{EqMaster}) could be associated with Markovianity: to predict the further dynamics of the system, it suffices to know the density operator of the system at a given instant of time, which intuitively implies the absence of memory. There is no unanimity in the literature regarding the generalization of the concept of Markovianity to the quantum case. Various approaches to the concept of Markovian quantum dynamics were discussed in Ref.~\cite{HierMark}. We will use to a definition that is equivalent to the definition of a classical Markovian random process in the classical case. It requires the knowledge of not only the reduced density operator of the system but also the correlation functions \cite{SinPetr}. Let $A_1,\ldots,A_n$ be operators in the space of the system. Denote by $\hat A_j(t)=\mathcal U_0(-t)(A\otimes I_R)$ the corresponding operators at time $t$ in the Heisenberg picture (hatted symbols are used for the Heisenberg picture, while the corresponding unhatted symbols are used for the interaction representation). Consider the correlation function
\begin{equation}\label{EqCorr}
\begin{split}
\langle
\hat A_n(t_n)\cdots\hat A_1(t_1)
\rangle
&=
\Tr\hat A_n(t_n)\cdots\hat A_1(t_1)\mathcal R\rho_S(0)
\\&=
\Tr A_n\mathcal U(t_{n}-t_{n-1}) A_{n-1}\cdots
A_1\mathcal U(t_1)\mathcal R\rho_S(0)\\
&=
\Tr A_n\mathcal U(t_{n}-t_{n-1}) A_{n-1}\cdots
\mathcal U(t_2-t_1)A_1\mathcal R\Phi(t_1)\rho_S(0),
\end{split}
\end{equation}
where we used equality (\ref{EqCommut}) as well as the convention that the superoperator (just as the operator of taking the trace) acts on the whole expression to the right of it. We say that the quantum process defined in our case by the semigroup generator $\mathcal G$ and the recovery operator $\mathcal R$ is Markovian if the correlation function can be expressed in terms of the semigroup $\Phi$:
\begin{equation}\label{EqCorrMark}
\langle
\hat A_n(t_n)\cdots\hat A_1(t_1)
\rangle
=
\Tr A_n\Phi(t_{n}-t_{n-1}) A_{n-1}\cdots
\Phi(t_2-t_1)A_1\Phi(t_1)\rho_S(0).
\end{equation}
The corresponding statement for pair correlators ($n=2$) is called the quantum regression theorem \cite{Lax,DumkeCorr}.

If we restrict ourselves to the zeroth order of the recovery operator in Eq.~(\ref{EqCorr}), i.e., if we replace $\mathcal R$ by $\mathcal R_0$ (which, as we have seen, corresponds to a second-order generator), then the validity of Eq.~(\ref{EqCorrMark})  can be easily proved due to the equality 
\begin{equation}\label{EqR0prod}
A\mathcal R_0 B=\mathcal R_0(AB)=AB\otimes\rho_R^{\rm ref}
\end{equation}
for any operators $A$ and $B$ in the space $\mathcal H_S$. Successively applying equalities (\ref{EqCorr}) and (\ref{EqR0prod}) in Eq.~(\ref{EqCommut}), we arrive at Eq.~(\ref{EqCorrMark}).
 
However, equality (\ref{EqR0prod}) is valid only for the zeroth order of the recovery operator. It fails in the general case. It is easy to see from Eq.~(\ref{EqR1}) that this equality fails to hold even for $\mathcal R_1$; moreover, in the general case,  $A\mathcal R_1B$ does not have the form $\mathcal RC$ for any operator $C$. Therefore, $\mathcal U(t_2-t_1)A_1\mathcal R\Phi(t_1)\rho_S(0)$ cannot be in general described by the semigroup $\Phi$: the argument of the superoperator $\mathcal U(t_2-t_1)$ does not have the form $\mathcal R C$. Thus, the quantum dynamics is Markovian only in the zero-order approximation of the recovery operator and in the second-order approximation of the generator. Since higher order corrections to the generator require nonzero terms in the expansion of the operator $\mathcal R$, we can conclude that the dynamics of the system in the regime of weak coupling to the reservoir, which is described by the generator with higher order corrections, is non-Markovian in the indicated sense.

\section{Discussion and conclusions}

We have demonstrated that the Bogoliubov method allows one to obtain quantum master equations for open quantum systems in an arbitrary order of perturbation theory with respect to a small parameter; in the present paper we took the coupling constant between the system and reservoir as such a parameter. We have proved that under certain assumptions the expansion terms of all orders for the generator are well defined. 

The method naturally considers initially correlated system-reservoir states arising from the previous system-reservoir dynamics.  Namely, the recovery operator $\mathcal R$ (\ref{EqR}) recovers a correlated system-reservoir state for a given reduced state of the system. A product (uncorrelated) state is just the zeroth-order approximation, see Eqs.~(\ref{EqRSer}) and (\ref{EqR0}), while corrections contain correlations, see Eqs.~(\ref{EqR1}), (\ref{EqR2}), and (\ref{EqR3}) for the first-, second-, and the third-order corrections, respectively. 

Let us mention some open questions. As we said above, the Redfield generator $\mathcal G_2$ can be transformed into a generator of the GKLS form by means of the secular or a more general approximation. It would be interesting to find out whether the corrections (\ref{EqGr}) and (\ref{EqGrInt}) to the Redfield equation can be reduced to the GKLS form by similar approximations. In  Ref.~\cite{Taepra}, it is shown that, for a particular exactly solvable model with Hamiltonian expressed in the rotating wave approximation, both the Redfield equation and all corrections to it have the GKLS form. From the other side, from Result~3 of Ref.~\cite{MerkliRev}, it might be conjectured that, in general, it is impossible to increase the accuracy of the steady-state (with respect to the zeroth order) maintaining the same or greater accuracy for the coherences (the part of the system density operator non-commuting with $H_S$) on intermediate times with a time-independent GKLS generator.

The main open problem is to rigorously substantiate the derivation presented above, especially assumptions (\ref{EqR}) and (\ref{EqBoundCond}). More precisely, we can speak of two aspects of this problem. The first is to rigorously substantiate the existence of such solutions, i.e., to prove that a solution of the Cauchy problem for the Liouville equation with the initial condition of the form (\ref{EqR}) has the form constructed here. The second aspect is to prove that even if the initial state does not have the form (\ref{EqR}) (which, as we said, is not rare: even the most frequently used factorized state (\ref{EqR0}) does not have the form  (\ref{EqR}) if we take into account orders higher than zero) but belongs to some wider class of ``physically interesting'' initial states, then, in a short time, it approaches a state of the form (\ref{EqR}). This would imply that the solutions constructed are in a sense exhaustive for physically interesting initial states. One can also check the validity of assumptions (\ref{EqR}) and (\ref{EqBoundCond}) for exactly solvable models, for example, for models solvable by the pseudomode method \cite{Taepra,Garraway1,Garraway2,Taetr,TaeLob}.

Finally, it would be interesting to consider the so called mean force Gibbs state from the viewpoint of the presented approach. Namely, consider a thermal reservoir with the inverse temperature $\beta$ and the system-reservoir Gibbs state $\rho_{SR,\beta}=Z^{-1}e^{-\beta H}$, where $Z=\Tr e^{-\beta H}$ (strictly speaking, $Z$ is ill-defined and $\rho_{SR,\beta}$ is not a density operator but a positive functional on the algebra of observables). It is stationary for the system-reservoir dynamics. The corresponding reduced state of the system $\rho_{S,\beta}=\Tr_R[Z^{-1}e^{-\beta H}]$ is called the mean force Gibbs state \cite{CresserAnders}. It appears that
\begin{equation}\label{EqMFG}
\mathcal R\rho_{S,\beta}=\rho_{SR,\beta}.
\end{equation}
Indeed, from Eq.~(\ref{EqCommut}), we see that $\mathcal R\rho_S$ is stationary for the exact system-reservoir dynamics whenever $\rho_S$ is stationary for $\Phi(t)$. So, if the joint system-reservoir state relaxes to the Gibbs state $\rho_{SR,\beta}$ (which is true under certain conditions \cite{Bach,Frohlich}), then $\Phi(t)$ has a unique stationary state $\rho_S^*$ and $\mathcal R\rho_{S}^*=\rho_{SR,\beta}$. Moreover, $\rho_S^*$ coincides with $\rho_{S,\beta}$ in view of Eq.~(\ref{EqTrRR}). 

Note that the mean force Gibbs state is not stationary for the map
\begin{equation}\label{EqProdEvol}
\rho_S\mapsto 
\Tr_R[e^{-iHt}(\rho_S\otimes\rho_R^{\rm ref})e^{iHt}]
\end{equation}
because the product state destroys the system-bath correlations, but is stationary for the map
\begin{equation}\label{EqRevol}
\rho_S\mapsto \Tr_R[e^{-iHt}(\mathcal R\rho_S)e^{iHt}]\equiv \Phi(t)\rho_S.
\end{equation}

\textbf{Acknowledgments.} I am grateful to Alexander Teretenkov for important remarks and references. I would like to dedicate this paper to the 75th anniversary of my teacher Prof. Igor Vasil'evich Volovich. This work is supported by the Russian Science Foundation under grant 17-71-20154.

\end{document}